\documentclass[letterpaper,onecolumn,11pt,accepted=2024-01-30]{quantumarticle}
\pdfoutput=1
\usepackage[utf8]{inputenc}
\usepackage[english]{babel}
\usepackage[T1]{fontenc}

\usepackage{tikz}
\usepackage{lipsum}

\usepackage[margin=0.75in]{geometry}
\usepackage[utf8]{inputenc}
\usepackage[english]{babel}
\usepackage[T1]{fontenc}
\usepackage{amsmath,amssymb}
\usepackage{graphicx}           
\usepackage{relsize}
\usepackage{tcolorbox}
\usepackage{amsthm}
\usepackage{thm-restate}
\usepackage{braket}
\usepackage{csquotes}

\usepackage{soul}
\usepackage{mathtools}

\usepackage{cite}
\usepackage[figurename=Fig.,font=footnotesize]{caption}

\theoremstyle{definition}
\usepackage{float}
\makeatletter
\newlength\min@xx

\makeatother

\usepackage{hyperref}
\hypersetup{
    colorlinks=true,
    linkcolor=blue,
    citecolor=magenta,
    filecolor=magenta,      
    urlcolor=blue,
    pdftitle={},
    pdfpagemode=FullScreen,
    }

\usepackage{lipsum}

\newtheorem{theorem}{Theorem}
\newtheorem{lemma}{Lemma}
\newtheorem{corollary}{Corollary}
\newtheorem{definition}{Definition}
\newtheorem*{remark}{Remark}

\newtheorem{example}{Example} 

\def\ket#1{| #1 \rangle}
\def\bra#1{\langle #1 |}
\def\kb#1#2{|#1\rangle\!\langle #2 |}

\title{Stabilizer Formalism for Operator Algebra Quantum Error Correction}

\author{Guillaume Dauphinais}
\affiliation{Xanadu, Toronto, ON M5G 2C8, Canada}

\author{David W. Kribs}
\affiliation{Xanadu, Toronto, ON M5G 2C8, Canada}
\affiliation{Department of Mathematics \& Statistics, University of Guelph, Guelph, ON N1G 2W1, Canada}
\email{dkribs@uoguelph.ca}

\author{Michael Vasmer}
\affiliation{Xanadu, Toronto, ON M5G 2C8, Canada}
\affiliation{Perimeter Institute for Theoretical Physics, Waterloo, ON N2L 2Y5, Canada}
\affiliation{Institute for Quantum Computing, University of Waterloo, Waterloo, ON N2L 3G1, Canada}

\begin{document}

\maketitle

\begin{abstract}
We introduce a stabilizer formalism for the general quantum error correction framework called operator algebra quantum error correction (OAQEC), which generalizes Gottesman's formulation for traditional quantum error correcting codes (QEC) and Poulin's for operator quantum error correction and subsystem codes (OQEC). The construction generates hybrid classical-quantum stabilizer codes and we formulate a theorem that fully characterizes the Pauli errors that are correctable for a given code, generalizing the fundamental theorems for the QEC and OQEC stabilizer formalisms. We discover hybrid versions of the Bacon-Shor subsystem codes motivated by the formalism, and we apply the theorem to derive a result that gives the distance of such codes.  
We show how some recent hybrid subspace code constructions are captured by the formalism, and we also indicate how it extends to qudits. 
\end{abstract}

\section{Introduction}

Quantum error correction (QEC) is a central topic in quantum information science. Its origins as an independent field of study go back almost three decades \cite{PhysRevA.52.R2493,knill1997theory,PhysRevA.54.1098,ShorFaultTolerant}, and it now touches on almost every aspect of quantum information, ranging from theoretical to experimental investigations and in recent years as a key facet in the development of new quantum technologies \cite{PhysRevLett.116.230502,PhysRevA.94.012311,Bourassa2021blueprintscalable,Wang2022,Zhuang_2020,holzgrafe2019,panteleev2022goodcodes}. More recently, developments in QEC included the introduction of a unified approach, called `operator quantum error correction' (OQEC)  \cite{Kribs2005Unified,Kribs2006oqec} that brought together traditional QEC with passive notions such as decoherence-free subspaces and noiseless subsystems, and led to the advent of subsystem codes and advances in fault-tolerant quantum computing \cite{Poulin2005Stabilizer,Bacon2006Operator,aly2008subsystem,Bombin2015Gauge,Bombin2015Single,PhysRevX.11.031039,kubica2022SingleshotQuantumError,Hastings2021dynamically}. Subsequently, a further generalization was discovered, called `operator algebra quantum error correction' (OAQEC) \cite{Beny2007Generalization,Beny2007Quantum}, which additionally provided an approach for hybrid classical-quantum codes used for the simultaneous encoding of classical and quantum information, and for infinite-dimensional error correction \cite{kuperberg2003capacity,devetak2005capacity,hsieh2010entanglement,hsieh2010trading,beny2009quantum}. 
The following decade saw limited development of OAQEC theory, perhaps due to a paucity of initial applications.

The last few years have witnessed significant renewed interest in OAQEC, from at least three different but related directions. There have been advances in hybrid classical-quantum information coding theory and error correction \cite{grassl2017codes,li2020error,cao2021higher,nemec2021infinite} that fit into the OAQEC framework. Several small quantum error correcting codes and operations necessary as fundamental components of a scalable fault-tolerant quantum computer have been implemented experimentally \cite{https://doi.org/10.48550/arxiv.2009.11482,krinner2021,https://doi.org/10.48550/arxiv.2207.06431,Postler2022}. And in black hole theory, recent work~\cite{Verlinde2013,Almheiri2015,Harlow2017,Hayden2019,almheiri2018holographic,Penington2020,akers2019,kamal2019ryu,Akers2022,akers2022black} has reinterpreted the AdS/CFT correspondence using the language of quantum error correction. In particular, it was argued in~\cite{Harlow2017} that the full machinery of OAQEC is necessary to capture the relevant properties of AdS/CFT.

The stabilizer formalism \cite{gottesman1996class,gottesman1997stabilizer,PhysRevLett.78.405} introduced by Gottesman is a bedrock of QEC, providing a toolbox for the construction and characterization of correctable codes for Pauli error models. This formalism was generalized by Poulin \cite{Poulin2005Stabilizer} to the OQEC setting, giving a way to construct stabilizer subsystem codes and also a characterization of correctable subsystem codes for Pauli errors. The OQEC formalism further gave an appropriate framework in which to view the well-known Bacon-Shor subsystem codes \cite{Bacon2006Operator}, which have proved to be important in fault tolerant quantum computing. 

In this paper, we introduce a stabilizer formalism for OAQEC, which generalizes Gottesman's formulation for traditional QEC codes and Poulin's for OQEC subsystem codes. The codes constructed  include hybrid classical-quantum stabilizer codes, and motivated by this, we discover hybrid versions of the Bacon-Shor codes. We formulate a theorem that fully characterizes the Pauli errors that are correctable for a given stabilizer code, generalizing the fundamental theorems for QEC and OQEC, and we apply the theorem to calculate the distance of hybrid Bacon-Shor codes.
Further, we show how some recent hybrid subspace code constructions are captured by the formalism. We also show how it extends to the case of qudits and we present examples in that general context. 

This paper is organized as follows. Section~2 includes requisite background material. In Section~3 we give the main details of the formalism, and in Section~4 we formulate and prove the error correction theorem. We present some examples and applications in Section~5, including the hybrid Bacon-Shor codes and a theorem that gives the distance of such codes. Section~6 includes the extension of the formalism to qudits, and Section~7 includes concluding remarks.

\section{Preliminaries}

Given a fixed positive integer $n \geq 1$, let $\mathbb C^N$, with $N = 2^n$, be $N$-dimensional complex Hilbert space with a fixed orthonormal basis $\{ \ket{0}, \ldots , \ket{N-1}\}$, which alternatively can be identified with $(\mathbb C^2)^{\otimes n}$ and orthonormal basis $\{ \ket{i_1\cdots i_n} = \ket{i_1}\otimes \ldots \otimes \ket{i_n} \, : \, i_j = 0 , 1 \}$ via dyadic expansions.  
Further let $M_N = (M_2)^{\otimes n}$ be the set of $N \times N$ complex matrices, which can be viewed as the set of matrix representations of linear transformations $\mathcal B(\mathbb C^N)$ on $\mathbb C^N$ with respect to the basis $\{\ket{k}\}$, and let $\mathcal U(N)$ be the unitary group inside $M_N$. 

We let $\mathcal P_n$ be the usual $n$-qubit Pauli group; that is, the subgroup of $\mathcal U(N)$ generated by $n$-tensors of the single qubit bit flip and phase flip Pauli operators $X$, $Z$, and $iI$ (we shall write $I_m$ for the identity operator on $\mathbb C^m$, or just $I$ when the context is clear); that is, 
\[
X\ket{0} = \ket{1}, \,\, X\ket{1} = \ket{0} \quad \mathrm{and} \quad  Z\ket{0} = \ket{0}, \, \, Z\ket{1} = -\ket{1}, 
\]
and with the corresponding $n$-qubit operators defined as $X_1 = X \otimes (I^{\otimes (n-1)})$, $X_2 = I\otimes  X \otimes (I^{\otimes (n-2)})$, etc. 

Given a subgroup of unitary operators $\mathcal G$ inside $\mathcal B(\mathcal H)$, the set of (bounded linear) operators on a Hilbert space $\mathcal H$, we let  $\mathrm{Alg}(\mathcal G)$ denote the subalgebra of $\mathcal B(\mathcal H)$ generated by $\mathcal G$; in other words, the set of complex polynomials in the elements of $\mathcal G$.  When $\mathcal H$ is finite-dimensional, such an algebra $\mathcal A$ is a (unital) C$^*$-algebra \cite{davidson1996c,paulsen2002completely}, and hence from the structure theory of such algebras, it is unitarily equivalent to a direct sum of the form 
\[
\mathcal A \cong \bigoplus_k (I_{m_k} \otimes M_{n_k})
\]
for some positive integers $m_k, n_k$  with $\sum_k m_k n_k = \dim \mathcal H$. Associated with this unitary equivalence is a decomposition of the Hilbert space $\mathcal H$ as an orthogonal direct sum of subspaces each with its own tensor decomposition, $\mathcal H = \oplus_k (A_k \otimes B_k)$, in which the algebra itself decomposes as $\mathcal A = \oplus_k (I_{A_k} \otimes \mathcal B(B_k))$. Moreover, the set $\mathcal A^\prime$  (also an algebra) of all operators that commute with the algebra, the {\it commutant} of $\mathcal A$, is unitarily equivalent to 
\[
\mathcal A^\prime \cong \bigoplus_k (M_{m_k} \otimes I_{n_k}), 
\] 
which again is determined by the structure of the Hilbert space decomposition as ${\mathcal A}^\prime = \oplus_k (\mathcal B(A_k) \otimes I_{B_k})$. 

Open system quantum dynamics gives us {\it quantum channels}, which are completely positive trace-preserving linear maps $\mathcal E : \mathcal T(\mathcal H) \rightarrow \mathcal T(\mathcal H)$ on the set of trace class operators on $\mathcal H$ \cite{nielsen_chuang_2010,Holevo_2013,paulsen2002completely}. To each channel there is an associated dual map $\mathcal E^\dagger$ defined on $\mathcal B(\mathcal H)$ via the equation: $\mathrm{Tr}(\mathcal E^\dagger(X) \rho ) = \mathrm{Tr}(X \mathcal E(\rho))$. (Observe that $\mathcal E$ is trace-preserving exactly when $\mathcal E^\dagger$ is unital; $\mathcal E^\dagger(I)=I$.) Of course, in the finite-dimensional case the sets $\mathcal T(\mathcal H)$ and $\mathcal B(\mathcal H)$ coincide, but we will still use the different notation to distinguish between the quantum information flow direction under consideration; namely, the   Heisenberg and Schr\"odinger perspectives as discussed in the OAQEC context below, where we focus on correcting observables in the former picture and correcting states in the latter picture. 

Every channel $\mathcal E$ has operator-sum representations \cite{choi1975completely}, which are sets of `Choi-Kraus' operators $\{ E_k \}$ inside $\mathcal B(\mathcal H)$ such that $\mathcal E(\rho) = \sum_k E_k \rho E_k^\dagger$ for all $\rho \in \mathcal T(\mathcal H)$ and $\sum_k E_k^\dagger E_k = I$. In the quantum error context, channels are often referred to as {\it error} or {\it noise models}, and the implementation operators called {\it error operators}. Most importantly for the present work, the class of Pauli error models are central to quantum error correction, and are the subclass of {\it mixed unitary channels} on $\mathbb{C}^N$ of the form $\mathcal E(\rho) = \sum_k p_k U_k \rho U_k^\dagger$, where $U_k\in \mathcal P_n$ and the $p_k$ form a classical probability distribution.

\section{Hybrid Stabilizer Code Construction}

Hybrid classical-quantum codes are codes that can be used to simultaneously encode classical and quantum information. In terms of Hilbert space representations, they are codes with several mutually orthogonal subspaces, each of which carrying certain properties. In this section we construct hybrid codes that both fit into the OAQEC framework and generalize the codes from the original stabilizer formalism. For clarity, we have divided the presentation into four subsections, which include the three core notions and a motivating class of examples. 

\subsection{Stabilizer Subgroup and Code Subspace}

Let $\mathcal S$ be an abelian subgroup of $\mathcal P_n$ that does not contain $-I$, and suppose it has $s$ independent generators. As all elements of the Pauli group either commute or anti-commute up to some power of $iI$, and $\mathcal S$ does not contain the subgroup $\langle iI \rangle$ generated by $iI$, it is easy to see that the normalizer and centralizer of $\mathcal S$ inside $\mathcal P_n$ coincide; 
\[
\mathcal N(\mathcal S) = \{ g\in \mathcal P_n \, \, | \, \, g \mathcal S g^{-1} = {\mathcal S}  \} = 
\{ g\in \mathcal P_n \,\, | \, \, g h = h g \, \, \forall h \in \mathcal S \} = \mathcal Z(\mathcal S).  
\]

Let $C = C(\mathcal S)$ be the {\it stabilizer subspace} for $\mathcal S$, which is the subspace of $\mathbb{C}^N$ defined as the joint eigenvalue-1 eigenspace for $\mathcal S$; that is, 
\[
C = \mathrm{span} \{ \ket{\psi} \, : \, g \ket{\psi} = \ket{\psi} \,\, \forall g\in \mathcal S \}. 
\]
We will let $P$ denote the codespace projector for $C$, the orthogonal projection of $\mathbb C^N$ onto $C$. 
It is well known that $\dim C = 2^{n-s}$ (for instance see the motivating example below). 
The stabilizer subspace is the base code for an OAQEC stabilizer code, which will encode further structure as described below.

\subsection{Gauge Group and Logical Operations}

Let us first discuss some relevant operator theoretic notions. Given any element $g$ of $\mathcal N(\mathcal S) = \mathcal Z(\mathcal S)$, the subspace $C$ is a reducing subspace for $g$; that is, both the subspace and its orthogonal complement are invariant for $g$. Indeed, if $g$ commutes with every element of $\mathcal S$, then $g P = P g$ as $P$ is equal to a polynomial in the elements of $\mathcal S$, which follows from the joint spectral functional calculus for those elements (an explicit formula is given in Section~6). Hence, $gP = PgP$ and $gP^\perp = P^\perp g P^\perp$ where $P^\perp = I - P$, which are the invariant subspace conditions for $C$ and $C^\perp$ as operator relations. 
Observe that if $C$ is a reducing subspace for every operator in an algebra $\mathcal A$, then $\mathcal A P$ is a subalgebra of $\mathcal B(\mathbb{C}^N)$ which is fully supported on $C$. We will call $\mathcal A P = P \mathcal A = P \mathcal A P$ the `compression algebra' of $\mathcal A$ to $C$. In such a case, as a notational convenience to distinguish between that algebra and the corresponding algebra of operators restricted to $C$ (so a subalgebra of $\mathcal B(C)$), we shall write $\mathcal A |_C$ for the latter.   

We now turn to the subsystem structure generated by a stabilizer subspace. Our formulation here is a little more abstract than that of \cite{Poulin2005Stabilizer}, with an eye toward possible extensions of this formalism as noted in Section~7. Thus, suppose we can find subsets $\mathcal G_0$ and $\mathcal L_0$ of $\mathcal N(\mathcal S) = \mathcal Z(\mathcal S)$ with the following properties: 
\begin{itemize}
\item[$\bullet$] The compression algebra $\mathrm{Alg}(\mathcal G_0) P$, respectively $\mathrm{Alg}(\mathcal L_0)P$,  is unitarily equivalent to a full matrix algebra $M_{2^r}$ for some positive integer $r$, respectively to $M_{2^k}$ for some positive integer $k$. (The motivating example presented below shows how this arises through anti-commuting pairs of Pauli operators.)

\item[$\bullet$] The sets $\mathcal G_0$ and $\mathcal L_0$ are mutually commuting; $[g,L]=0$ for all $g\in \mathcal G_0$, $L\in \mathcal L_0$. 

\item[$\bullet$] The normalizer subgroup $\mathcal N(\mathcal S) $ is generated by $\mathcal S$, $iI$, $\mathcal G_0$, and $\mathcal L_0$. 
\end{itemize}

We will assume these sets are minimal with these properties (in particular, an element of the set cannot be obtained as a product of other elements). The group $\mathcal G$ defined as
\begin{equation} \label{eqn:def_gauge_group}
    \mathcal G = \langle \mathcal S, iI, \mathcal G_0 \rangle,
\end{equation}
is called the {\it gauge group} for the code, and the group
\begin{equation} \label{eqn:def_logical_group}
    \mathcal L = \langle \mathcal L_0, iI \rangle,
\end{equation}
is called the {\it logical group}. The conditions ensure the normalizer is determined by the gauge and logical groups via the group isomorphism $\mathcal N(\mathcal S) \times \langle i I \rangle \cong \mathcal G \times \mathcal L$. 
Choices of such subsets can be made using well-known properties of the Pauli group, using ($r$ and $k$ respectively) anti-commuting pairs of operators that mutually commute, and that commute with the other set and the stabilizers. 

The subsystem structure that these subgroups generate is given in the following result. This can be proved straightforwardly as a consequence of the above formulation, together with the structure theory of algebras and their commutants described in the previous section. 

\begin{lemma}\label{reducinglemma}
Let $C$ be a code subspace with gauge group $\mathcal G$ and logical group $\mathcal L$ as chosen above. Then $C$ is a reducing subspace for both $\mathcal G$ and $\mathcal L$, and $C$ decomposes as a tensor product of subsystems $C = A \otimes B$ with $A \cong (\mathbb C^2)^{\otimes r}$, $B \cong (\mathbb C^2)^{\otimes k}$, and $r+ k = n - s$, such that 
\[
\left\{ 
\begin{array}{rcl}
\mathrm{Alg} (\mathcal G)|_C & = & \mathcal B(A) \otimes I_B   \\ 
\mathrm{Alg} (\mathcal L)|_C & = & I_A \otimes \mathcal B(B) 
\end{array}
\right. ,
\]
where $\mathcal B(A) \cong M_{2^r}$ and $\mathcal B(B) \cong M_{2^k}$. 
\end{lemma}

Note that here the subsystem $B$ encodes the logical qubits of the code. 
Further observe that an empty gauge set $\mathcal G_0$ leads to a standard subspace code ($\dim A = 1$), whereas a nonempty $\mathcal G_0$ generates subsystem structure in the code (that is, when $\dim A >1$).

\subsection{Normalizer Cosets and Hybrid Code Sectors}

Let us now turn to the notion that generates hybrid codes. As a group theoretic observation that will be relevant below, first note that the left and right cosets of $\mathcal N(\mathcal S)$ inside $\mathcal P_n$ coincide; that is, $g \mathcal N(\mathcal S)  = \mathcal N(\mathcal S)  g$ for all $g\in \mathcal P_n$. This follows from the anti-commutation relations of $\mathcal P_n$ and the fact that $\mathcal N(\mathcal S) $ contains $\langle  i I \rangle$.  

Let $\mathcal T \subseteq \mathcal P_n$  be a maximal set of coset representatives for $\mathcal N(\mathcal S)$ inside $\mathcal P_n$ (a so-called coset {\it transversal} for $\mathcal N(\mathcal S)$ as a subgroup of $\mathcal P_n$), and without loss of generality assume $I\in \mathcal T$ is the representative for the normalizer itself. Then the full group is equal to the (disjoint) union $\mathcal P_n = \cup_{g\in \mathcal T} \, g \mathcal N(\mathcal S) $, and the cardinality of $\mathcal T$ is equal to $|\mathcal T| = | \mathcal P_n | / |\mathcal N(\mathcal S) | = 2^s$ (see the motivating example for an explicit calculation). 

Taking terminological motivation from other areas, such as the notion of `charge sectors' in the study of topological codes \cite{https://doi.org/10.48550/arxiv.hep-th/9511201,Kitaev2003Anyons}, we shall use the term {\it code sector} to refer to the (quantum) code defined  by a given $T\in \mathcal T$ and the elements that define the base code: $\mathcal S$, $\mathcal L$, $\mathcal G$. Specifically, the code sector for $T$ is defined by the collection of operators given by the sets $T {\mathcal S} T^{-1}$, $T {\mathcal L} T^{-1}$, $T {\mathcal G} T^{-1}$, and then the associated codespace $T C$. 

The key observation concerning normalizer cosets in this setting is the following, which shows that the subgroup and coset  structure induces orthogonality at the Hilbert space level. 

\begin{lemma}\label{cosetlemma}
Let $\mathcal S$ be an abelian subgroup of $\mathcal P_n$ that does not contain $- I$, and let $C$ be its stabilizer subspace. If $\mathcal T$ is a selection of coset representatives for $\mathcal N(\mathcal S)$ inside $\mathcal P_n$, then for all $g_1,g_2\in \mathcal T$ with $g_1\neq g_2$ we have 
\[
P g_1^{-1} g_2 P= 0, 
\]
where $P$ is the orthogonal projection of $\mathbb{C}^N$ onto $C$; in other words, $g_1\ket{\psi_1}$ is orthogonal to $g_2\ket{\psi_2}$ for any choice of states $\ket{\psi_1}, \ket{\psi_2} \in C$.  
\end{lemma}

\begin{proof}
As $g_1$ and $g_2$ are representatives from different cosets, we have $g := g_1^{-1} g_2 \notin \mathcal N(\mathcal S)$. Since the normalizer coincides with $\mathcal Z(\mathcal S)$, and all elements of $\mathcal P_n$ commute modulo a power of $iI$, it follows that there is some $E\in \mathcal S$ and $z\in \mathbb C$, with $|z|=1$ and $z\neq 1$, such that 
\[
gE = z Eg.
\]
(In fact here in the qubit $d=2$ case we must have $z=-1$. We keep the logical statement more general as we will make use of this argument in the qudit case.) 
Note that $EP = P$ from the definition of $C$ and because $E\in \mathcal S$. Further, $EP = PE$ as $P$ is equal to the product of polynomials in elements of $\mathcal S$ from spectral theory functional calculus as discussed above. Thus we have, 
\[
PgP = P g E P = P(z E g)P = z P E g P = z EP gP = z PgP,  
\] 
and so $(1-z) PgP = 0$. But $z\neq 1$, and hence $P g P =0$ as required. 
\end{proof}

\begin{remark}
We thus have stabilizer codes that generalize both the original (subspace) setting of Gottesman, which is captured with the singleton coset representative subset ($\mathcal T_0 = \{ I \}$) and abelian gauge group (with empty set $\mathcal G_0 = \emptyset$), and then the OQEC (subsystem) setting of Poulin, which is captured with the singleton coset representative ($\mathcal T_0 = \{ I \}$) subset and non-trivial gauge group ($\mathcal G \neq \emptyset$). 

\begin{figure}
    \centering
    \includegraphics[width=0.65\linewidth]{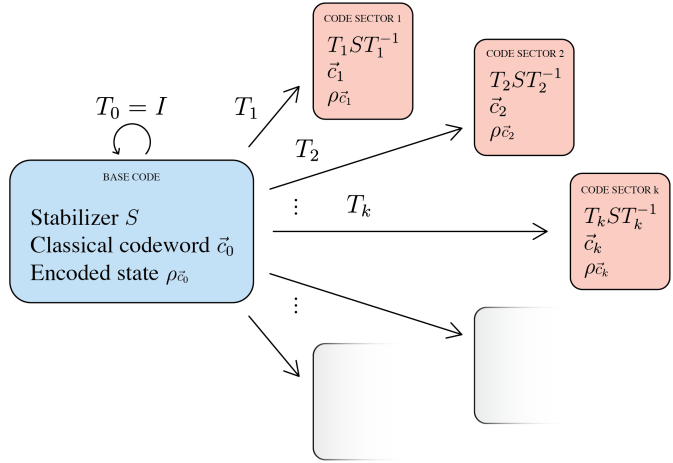}
    \caption{Hybrid stabilizer code illustration. The coset representatives $\mathcal T_0 = \{T_0 = I, T_1, T_2, \dots, T_k \} \subset \mathcal T$ define different code sectors, each characterized by a classical codeword represented by a bit-string $\vec c_i$, 
    with $0 \leq i \leq k$, 
    its associated encoded state $\rho_{\vec c_i}$ and sets $T_i \mathcal S T_i^{-1}$, $T_i \mathcal L T_i^{-1}$, $T_i \mathcal G T_i^{-1}$ (only the stabilizers are shown for brevity). 
    Note that here $\mathcal T_0$ is a strict subset of $\mathcal T$.
    }
    \label{fig:code-sectors}
\end{figure}

Moreover, any code defined by a subset $\mathcal T_0 \subseteq \mathcal T$ with $|\mathcal T_0|>1$ will be a hybrid classical-quantum code, which will have a sub{\it space} base code (formally $C = A \otimes B$ with $A = \mathbb C$) when the gauge group is abelian and a sub{\it system} base code ($C=A\otimes B$ with $\dim A > 1$) otherwise. The size of the subset $\mathcal T_0$ determines the number of  `classical addresses' associated with the hybrid code, as is illustrated in Fig.~\ref{fig:code-sectors}.  For instance,  by Lemma~\ref{cosetlemma}, any $g\notin \mathcal N(\mathcal S)$ gives a coset $g \mathcal N(\mathcal S)$ for which the subspace $g C$ is orthogonal to $C$, and hence it defines a 1-bit, $k$-qubit hybrid code (which may have further subsystem structure when the gauge group is non-abelian).   
\end{remark}

\subsection{Motivating Class of Examples}\label{motivating_section}

We can build upon the motivating example from the original stabilizer formalism to illustrate the various concepts discussed above for a relatively simple class of examples. The following operators are defined on $n$-qubit Hilbert space. 

\begin{itemize}
\item[$\bullet$] Let $s \leq n$ be a fixed positive integer and let $\mathcal S = \langle Z_1, \ldots , Z_s \rangle$; the subgroup of $\mathcal P_n$ generated by phase flip operators on the first $s$ qubits. 
Then 
\[
C = C(\mathcal S) = \mathrm{span} \big\{ \ket{ \underbrace{0 \cdots 0}_{s} \, i_1 \cdots i_{n-s}  } \, : \, i_j = 0,1 \big\}, 
\]
and $\dim C = 2^{n-s}$ so that $C$ can encode $n-s$ qubits. 

\item[$\bullet$] Let $r$ be a fixed integer with $0 \leq r \leq n-s$, and let $\mathcal G_0$ be the set of $r$ pairs of Pauli operators acting on qubits $s+1$ to $r+s$: 
\[
\mathcal G_0 = \big\{  X_i, Z_i \, : \, s+1 \leq i \leq r+s \big\}. 
\]
Then the gauge group $\mathcal G$ is generated by $\mathcal S$, $iI$, and $\mathcal G_0$, and includes the full Pauli subgroup of operators acting non-trivially on the $r$ `gauge qubits'.  

\item[$\bullet$] Let $k = n - s - r$,  and let $\mathcal L_0$ be the set of $k$ pairs of Pauli operators acting on qubits $r+s+1$ to $n$:
 \[
\mathcal L_0 = \big\{  X_i, Z_i \, : \, r+s+1 \leq i \leq n \big\}. 
\]
Then the logical group $\mathcal L$ is the group generated by $\mathcal L_0$ and $iI$, and includes the full Pauli subgroup of operators acting non-trivially on the $k$ `logical qubits'. 

\item[$\bullet$] The normalizer $\mathcal N(\mathcal S) = \mathcal Z(\mathcal S)$ for $\mathcal S$ inside $\mathcal P_n$ in this case is given by the following set of operators: 
\[
\mathcal N(\mathcal S) = \Big\{   i^c \, \cdot \, Z_1^{b_1} \cdots Z_s^{b_s} \, \cdot \,  X_{s+1}^{a_{s+1}} Z_{s+1}^{b_{s+1}} \cdots    X_{n}^{a_{n}} Z_{n}^{b_{n}} \, : \, 0\leq c \leq 3, \,\, 0 \leq a_j, b_j \leq 1   \Big\}. 
\]
\end{itemize}

The size of the normalizer here is thus $|\mathcal N(\mathcal S)| = 4 \cdot 2^s \cdot 4^{n-s} = 2^{2 + 2n -s}$. The full Pauli group $\mathcal P_n$ has $4^{n+1}$ elements (as every element can be uniquely written in the form $i^c  X_{1}^{a_{1}} Z_{1}^{b_{1}} \cdots    X_{n}^{a_{n}} Z_{n}^{b_{n}}$), and hence the number of normalizer cosets is given by, 
\[
|\mathcal P_n| / |\mathcal N(\mathcal S)| = 2^{2  + 2n - (2 + 2n -s)}  = 2^s. 
\]

Observe that each of the operators $X_j$, $1\leq j \leq s$, do not belong to $\mathcal N(\mathcal S)$. Hence we can take as a set of canonical coset representatives, the transversal given by the following $2^s$-element set: 
\[
\mathcal T = \big\{ X_1^{a_1} \cdots X_s^{a_s} \, : \, 0 \leq a_j \leq 1 \big\}.  
\]
As a caveat, however, we note that there are many other choices of coset representatives, which could have different algebraic properties as it relates to the code generators. As a simple example, note that $X_i$, with $1 \leq i \leq s$, and $X_i N$, for some fixed $N\in \mathcal N(\mathcal S)$, generate the same coset, and so in particular a transversal need not consist entirely of mutually commuting operators, or even operators that commute with the gauge and logical operators. 

Regarding the Hilbert space decomposition generated by this example, notice that the gauge and logical operators induce a tensor decomposition for the base code subspace $C = A \otimes B$, where $A \cong (\mathbb C^2)^{\otimes r}$,  $B \cong (\mathbb C^2)^{\otimes k}$, and this tensor structure naturally translates to any of the subspaces $T C$, for $T\in \mathcal T$, as $T$ is unitary. (Recall here that $n-s = r+k$, and the base code encodes $k$ logical and $r$ gauge qubits.)  The subspace $C$ is easily seen to be invariant for each of the gauge and logical operators, and evidently for every $A\in \mathcal G_0$ and $B\in \mathcal L_0$, there are operators $A_1\in\mathcal B(A)$ and $B_1\in\mathcal B(B)$ such that 
\[
\left\{ 
\begin{array}{rcl}
A|_C & = & A_1 \otimes I_B   \\ 
B|_C & = & I_A \otimes B_1 
\end{array}
\right. ,  
\]
which is all true in general by Lemma~\ref{reducinglemma}. 
Given a (non-trivial) subset of coset representatives $\mathcal T_0 \subseteq \mathcal T$, the subspaces $T C$, $T\in\mathcal T_0$, are mutually orthogonal (in general this is true by Lemma~\ref{cosetlemma}) and the corresponding subspace for the hybrid code is $C_{\mathcal T_0} = \oplus_{T\in\mathcal T_0} T C$.

In Sections~5 and 6 we will give further examples of the general hybrid code construction and discuss them in detail. We next we turn to an analysis of what are the possible errors that a given hybrid stabilizer code can protect against. 

\section{Error Correction Theorem}

The code construction above thus defines codes $C = C(\mathcal S, \mathcal G_0, \mathcal L_0, \mathcal T_0)$, determined by, respectively, choices of stabilizer subgroup, gauge and logical  operators, and subset of coset representatives. We shall characterize what sets of Pauli errors are correctable for a given code, and in doing so, we establish a generalization of the fundamental theorems of \cite{gottesman1996class,gottesman1997stabilizer} and \cite{Poulin2005Stabilizer} to this setting.

We shall first recall the basic notions and relevant results from OAQEC \cite{Beny2007Generalization,Beny2007Quantum}, and then specify to the code framework formulated above. The starting point is the basic definition of OAQEC codes, which is most conveniently introduced in the Heisenberg picture.

\begin{definition}\label{correctable}
An algebra $\mathcal A \subseteq \mathcal B(\mathcal H)$ of operators on $\mathcal H$ is {\em correctable} for an error model $\mathcal E$ if there exists a channel $\mathcal R$ such that $\mathcal A$ is conserved by $\mathcal R \circ \mathcal E$ on states in $Q\mathcal H$ where $Q$ is the unit projection of $\mathcal A$; that is,
\begin{equation}\label{Heiseneqn}
Q(\mathcal R \circ \mathcal E)^\dagger(X)Q = X \quad \forall
X\in\mathcal A.
\end{equation}
\end{definition}

Given the unitary equivalence form for an algebra $\mathcal A \cong \oplus_i ( I_{m_i} \otimes M_{n_i})$, the unit element of $\mathcal A$ is the projection $Q$ in the algebra corresponding under the equivalence to $\oplus_i I_{m_i} \otimes I_{n_i}$. There is a more general notion of OAQEC code considered in \cite{Beny2007Generalization,Beny2007Quantum}, wherein the unit element of $\mathcal A$ is replaced by an arbitrary projection on the Hilbert space, and one can then consider correction with respect to states supported on the corresponding range subspace of the projection. But the notion of correctable we consider here allows us to unambiguously discuss `correction of an algebra', and is sufficient for our goal to extend the stabilizer formalism to this setting.

The corresponding Schr\"odinger picture description is given as follows: $\mathcal A$ is correctable for $\mathcal E$ if and only if there exists a channel $\mathcal R$ such that for any density operator $\rho = \sum_i \alpha_i (\tau_i \otimes \rho_i)$ with $\tau_i \in \mathcal T(A_i)$, $\rho_i \in \mathcal T(B_i)$, and nonnegative scalars $\sum_i \alpha_i=1$, there are density operators $\tau_i' \in \mathcal T(A_i)$ for which
\begin{equation}\label{Schroeqn}
(\mathcal R\circ\mathcal E)(\rho) = \sum_i \alpha_i \mathcal R
\bigl({\mathcal E \bigl({\tau_i \otimes \rho_i}\bigr)}\bigr) =
\sum_i \alpha_i ( \tau_i' \otimes \rho_i).
\end{equation}
From this perspective, one can see that each of the subsystems $B_i$ (with $\dim B_i >1$) can be used individually to encode quantum information that can be recovered. Moreover, an extra feature of such a code is that an arbitrary mixture of encoded states, one for each subsystem, can be simultaneously corrected by the same correction operation.

To generalize the main error correction theorem from previous stabilizer formalism settings, we need a description in terms of error operators. The following result from \cite{Beny2007Generalization,Beny2007Quantum} gives such a description, and below we formulate it in a style that we will use.

\begin{theorem}
\label{found} Let $\mathcal A$ be a subalgebra of $\mathcal B(\mathcal H)$ with unit projection $Q$. The following statements are equivalent:
\begin{enumerate}
\item $\mathcal A$ is correctable for $\mathcal E(\rho) = \sum_k E_k \rho E_k^\dagger$.  
\item $[Q E_k^\dagger E_l^{} Q, X] = 0$ for all $X \in \mathcal A$ and all $k,l$.
\end{enumerate}
\end{theorem}

It follows from Theorem~\ref{found}, using the structure of finite dimensional algebras and their commutants discussed above, that there is a correction operation $\mathcal R$ for which Eq.~(\ref{Schroeqn}) is satisfied if and only if for all $k,l$ there are operators $X_{kli} \in \mathcal B(A_i)$ such that
\begin{equation}\label{oaqeceqn}
Q E_k^\dagger E_l^{} Q = \sum_i  X_{kli} \otimes I_{B_i},
\end{equation}
where here the operators $X_{kli} \otimes I_{B_i}$ are understood to act on $A_i \otimes B_i$, and so the sum (when there is a sum with more than one term) is thus an orthogonal direct sum of operators. The case of a sum with a single term captures the well-known Knill-Laflamme error correction conditions \cite{knill1997theory} when $\dim A_1 = 1$ (and so $X_{kl1}$ are complex scalars), and the OQEC testable conditions \cite{Kribs2005Unified,Kribs2006oqec} when $\dim A_1 > 1$.

Now let us specify to our setting. Further notation will be introduced in the proof below, but we will note here that the algebras associated with the code constructions of the previous section in their unitary equivalence form satisfy $m_i = m_{j}$ and $n_i = n_{j}$ for any two pair of indices $i,j$. Moreover, by saying a code  $C = C(\mathcal S, \mathcal G_0, \mathcal L_0, \mathcal T_0)$ is correctable, we mean the algebra determined by it, as in the previous section and the discussion above, is OAQEC-correctable.

\begin{theorem} \label{errorcorrecthm}
A code $C = C(\mathcal S, \mathcal G_0, \mathcal L_0, \mathcal T_0)$, with $\mathcal T_0 = \{ g_i \}_i$,  is correctable for a set of error operators $\{E_k\} \subseteq \mathcal P_n$ if and only if for all $k,l$,  
\begin{equation}\label{stabformoaqeccond}
E_k^\dagger E_l \notin \Big(  \mathcal N(\mathcal S)\setminus \mathcal G \Big)  \bigcup \Big(  \bigcup_{i\neq j} g_i \mathcal N(\mathcal S) g_j^{-1}   \Big) .
\end{equation}
\end{theorem}

\begin{proof}
First note that for any $g\in \mathcal P_n$, we have equality of the following operator sets: 
\[
\mathcal N(\mathcal S)\setminus \mathcal G = g  \Big(  \mathcal N(\mathcal S)\setminus \mathcal G \Big) g^{-1},
\]
which follows from basic group properties, the anti-commutation relations, and the fact that $\mathcal G$ (and $\mathcal N(\mathcal S)$) includes the scalar operators. 
Next let us establish some notation. Recall that $P$ is the projection onto $C$, and for each $i$, let $P_i = g_i P g_i^{-1}$. This is the projection onto the subspace $g_i C:= A_i \otimes B_i$, which has subsystem tensor structure $A_i, B_i$ induced by that of $C = A_1 \otimes B_1$ and the unitary action of $g_i$. So $g_i (\ket{a}\ket{b})$ will define an orthonormal basis for the new subspace from a basis $\ket{a}\ket{b}$ for $C$, which gives a corresponding identification of operators in $g_i (\mathcal B(A)\otimes \mathcal B(B)) g_i^{-1}$ with operators in $\mathcal B(A_i)\otimes \mathcal B(B_i)$; in particular, for any $X_A\in\mathcal B(A)$ this maps $X_A \otimes I_B$ to $X_{A_i}\otimes I_{B_i}$ for some $X_{A_i}\in \mathcal B(A_i)$.

By Lemma~\ref{cosetlemma}, the projections $P_i$ project onto mutually orthogonal subspaces and hence we define the projection $Q= \sum_i P_i$ to be the (orthogonal direct) sum of the $P_i$, and with $P_1 = P$. The algebra in the background to be corrected and defined by the code is $\mathcal A = \oplus_i (I_{A_i} \otimes \mathcal B(B_i) )$, which has $Q$ as its unit projection, and the error correction conditions we make use of are those of Eq.~(\ref{oaqeceqn}).

Throughout the proof we will let $E:=E_k^\dagger E_l$ for a fixed pair $k,l$. We shall first prove the `if' direction of the result. If $E$ is not in any of the sets in Eq.~(\ref{stabformoaqeccond}), then in particular for each $i$, we have $E \notin \mathcal N(\mathcal S) \setminus \mathcal G = g_i ( \mathcal N(\mathcal S) \setminus \mathcal G ) g_i^{-1}$ and so $g_i^{-1} E g_i \notin \mathcal N(\mathcal S) \setminus \mathcal G$. Thus by the OQEC special case of the theorem above (or the Knill-Laflamme theorem when $\mathcal G = \emptyset$), we have for some operator $X_A \in \mathcal B(A)$,
\[
P ( g_i^{-1} E g_i ) P   =    X_A \otimes I_B  ,
\]
and hence for some operator $X_{A_i}\in \mathcal B(A_i)$,
\[
P_i E P_i = g_i ( P g_i^{-1} E g_i P ) g_i^{-1} 
=    g_i (   X_A \otimes I_B   ) g_i^{-1}  
=  X_{A_i} \otimes I_{B_i} ,
\]
with the first and last equalities following from the definition of $P_i$ and the induced tensor structure on $g_i C$ from the decomposition $C=A \otimes B$ and the action of $g_i$ discussed above.

For the off-diagonal blocks, choose $i\neq j$ and assume $E\notin g_i \mathcal N(\mathcal S) g_j^{-1}$. Then we have $g_i^{-1} E g_j \notin \mathcal N(\mathcal S)$, and we claim that $P   g_i^{-1} E g_j P=0$ through what has become a standard stabilizer formalism type argument. Indeed, in general if $F\notin \mathcal N(\mathcal S) = \mathcal Z(\mathcal S)$, then there is some $g\in \mathcal S$ and $z\in \mathbb C$ with $|z|=1$ and $z\neq 1$ such that $F g = z g F$. (As above, necessarily $z = -1$ here, but we keep the argument general as we will extend the result to qudits below.) Also, as $g\in \mathcal S$ and from the construction of $C$, we have $P = gP = Pg$ (the latter following from spectral theory since $P$ is a polynomial in the elements of the abelian subgroup $\mathcal S$). Hence,
\[
PFP = PF gP = P ( z gF) P = z (PgFP) = z PFP,
\]
and so $PFP =0$. Thus we have, for all $i \neq j$, 
\[
P_i E P_j = g_i (P g_i^{-1} E g_j P ) g_j^{-1} = 0.
\]

It follows then that $QEQ = \sum_{i,j} P_i E P_j = \sum_i X_{A_i} \otimes I_{B_i}$, and hence each of the operators $Q E_k^\dagger E_l Q$ satisfy the form given in Eq.~(\ref{oaqeceqn}). Thus, $C$ is correctable for the set of error operators $\{E_k\}$.

Conversely, for the `only if' direction of the proof, suppose $E$ satisfies $Q E Q  = \sum_i X_{A_i} \otimes I_{B_i}$ for some operators $X_{A_i}\in \mathcal B(A_i)$. Then for a fixed $i$, using this form for $QEQ$ and that $P_i Q = P_i = QP_i$, we have
\[
X_{A_i} \otimes I_{B_i} = P_i QEQ P_i = P_i E P_i = g_i (P g_i^{-1} E g_i P ) g_i^{-1} .
\]
Hence for some $X_{A_1} \in \mathcal B(A_1)$, we have 
\[
P g_i^{-1} E g_i P   = g_i^{-1} (X_{A_i} \otimes I_{B_i}) g_i =  X_{A_1} \otimes I_{B_1} .
\]
It follows from OQEC (and QEC when $\mathcal G = \emptyset$) that $g_i^{-1} E g_i \notin \mathcal N(\mathcal S) \setminus \mathcal G$, and so $E  \notin g_i (\mathcal N(\mathcal S) \setminus \mathcal G) g_i^{-1} = \mathcal N(\mathcal S) \setminus \mathcal G$.

Finally, fix a pair $i\neq j$, and observe from the OAQEC correctable condition in Eq.~(\ref{oaqeceqn}) that the $i,j$ off-diagonal block in the block diagonal decomposition determined by $Q$ must be zero; that is, $P_i E P_j = 0$. Thus we have,
\[
g_i P g_i^{-1} E g_j P g_j^{-1} = P_i E P_j = 0,
\]
and since $g_i, g_j$ is unitary, in fact we have $P g_i^{-1} E g_j P = 0$. We want to conclude that $g_i^{-1} E g_j \notin \mathcal N(\mathcal S)$. Suppose instead we had $F:= g_i^{-1} E g_j \in \mathcal N(\mathcal S) = \mathcal Z(\mathcal S)$. Then $F P = P F$, from spectral theory and the construction of $P$, and so $C$, the range subspace of $P$, is a reducing subspace for $F$. Hence, $P^\perp F P = P^\perp P F = 0$, and so $FP = P^\perp FP + PFP = 0$.  But $F$ is a unitary operator, and so $F$ restricted to $C$ must be a norm-preserving map. This contradicts the fact that $FP=0$, and thus we must have $g_i^{-1} E g_j \notin \mathcal N(\mathcal S)$ as required, and this completes the proof.
\end{proof}

\begin{remark} 
Conceptually, not belonging to the first set of the theorem statement ensures the individual codes (which are OQEC subsystem codes when the subspace has a tensor decomposition) are correctable, and with the needed orthogonality for the multiple codes given by the choice of coset representatives. Not belonging to the operator sets in the second union encapsulates the joint hybrid classical-quantum correctable code conditions in the OAQEC framework. A larger subset of coset representatives corresponds to more sectors and a larger hybrid code. In particular, more sectors means generally larger operator sets in the theorem statements, which in turn makes it more difficult for errors to not belong to the sets, and hence smaller sets of correctable errors. These notions will be explored more in the examples below. 
\end{remark}

Regarding correctable sets of errors for the class of codes discussed in Section~\ref{motivating_section}, Theorem~\ref{errorcorrecthm} gives a full characterization of the possible correctable errors for any given coset subset $\mathcal T_0$. As a simple example, consider the case with the two operators $\mathcal T_0 = \{ I, X_1\}$ from the transversal $\mathcal T$ above. Here there are two operator sets that the error operator products $E_k^\dagger E_l$ cannot belong to: (i) $\mathcal N(\mathcal S) \setminus \mathcal G$; and (ii) $X_1 \mathcal N(\mathcal S) = \mathcal N(\mathcal S) X_1$. From the normalizer and gauge structures above, the first set consists of all elements in the normalizer with scalar multiples of the identity on qubits $s+1$ through to $r+s$. This set encapsulates the (quantum) error correction conditions for the two quantum codes defined by this code. The second set is simply all elements of the normalizer multiplied by $X_1$, and it corresponds to the cross terms that govern whether the code is hybrid correctable. 

Further, an example of a set of correctable errors for this code, one could take a subset of $\{ I, X_2, \ldots , X_s \}$. Any pairwise products of these operators do not belong to the two sets above and hence are correctable. (In fact, in general, any set of coset representatives not used to define the hybrid code will be correctable errors.) Sets of errors are not correctable by the theorem if any product of two of them belongs to either of the two sets. Thus, any error operator product of the form $X_1 N$, with $N\in \mathcal N(\mathcal S)$, would disrupt any hybrid correction for the error model, whereas any product belonging to the first set would prevent the individual quantum codes from being corrected.

\section{Further Examples and Applications}

\subsection{Hybrid Subspace Codes}

When the gauge group is abelian ($\mathcal G_0 = \emptyset$), the codes constructed above have no subsystem structure and are subspaces. 
Further, when additionally the coset representative set is non-trivial ($\{ I \} \subsetneq \mathcal T_0$), the codes generated by the formalism are `hybrid subspace codes’. 
From the OAQEC perspective, hybrid subspace codes are those associated with algebras $\mathcal A$ that are unitarily equivalent to a direct sum of full matrix algebras; i.e., of the form $\mathcal A \cong \oplus_{k=1}^M M_{n_k}$ for some positive integers $n_k$ (and $|\mathcal T_0|=M$ in our notation above). These are precisely the algebras with an abelian commutant, $\mathcal A^\prime \cong \oplus_{k=1}^M \mathbb{C} I_{n_k}$. Each summand thus can be used to encode quantum information (when $n_k >1$) as a traditional quantum (subspace) code, and overall the collection of codes defined by $\mathcal A$ make up a hybrid subspace code that can be corrected for error sets given by Theorem~\ref{errorcorrecthm}.   

The testable conditions of Eq.~(\ref{oaqeceqn}) take on a particularly transparent form in this case. As before, let $Q = \sum_i P_i$ be the unit projection of $\mathcal A$, with $P_i$ the projection onto the $i$th matrix block of $\mathcal A$. Then the code is correctable for a set of error operators $\{ E_k \}$  if and only if there are complex scalars $\lambda_{kl}^{(i)}$ such that for all $k, l$, 
\begin{equation}\label{subspacecond}
Q E_k^\dagger E_l Q = \sum_i \lambda_{kl}^{(i)} P_i . 
\end{equation}
These conditions can be cast into vector state form (as discussed in \cite{majidy2018unification}) as follows: Given $1 \leq i \leq M$, choose orthonormal states $\{ \ket{\psi_{i,j}}\}_{j=1}^{n_i}$ such that $P_i = \sum_j \kb{\psi_{ij}}{\psi_{ij}}$. We then have 
\[
\bra{\psi_{i_1 j_1}} E_k^\dagger E_l \ket{\psi_{i_2 j_2}} = \bra{\psi_{i_1 j_1}} Q E_k^\dagger E_l Q \ket{\psi_{i_2 j_2}}  
= \sum_i \lambda_{kl}^{(i)} \bra{\psi_{i_1 j_1}} P_i \ket{\psi_{i_2 j_2}} 
= \lambda_{kl}^{(i_1)} \delta_{i_1 i_2} \delta_{j_1 j_2} . 
\]
One can reverse this argument to observe that Eq.~(\ref{subspacecond}) is equivalent to the orthogonality conditions 
\begin{equation}\label{subspaceveccond}
\bra{\psi_{i_1 j_1}} E_k^\dagger E_l \ket{\psi_{i_2 j_2}} 
= \lambda_{kl}^{(i_1)} \delta_{i_1 i_2} \delta_{j_1 j_2} ,
\end{equation}
for any choice of orthonormal basis states $\ket{\psi_{ij}}$ for the range subspaces of the $P_i$. 

Recently, a distinguished special case of hybrid subspace codes was considered in \cite{grassl2017codes}.
In OAQEC language (with notation used in \cite{grassl2017codes}), the algebras of focus there are of the form $\mathcal A \cong \oplus_{\nu =1}^M M_K$; i.e., the direct sum of $M$ copies of $K\times K$ complex matrices. The full code subspace is thus an orthogonal direct sum $\oplus_{\nu =1}^M C^{(\nu)}$ of $K$-dimensional subspaces, and there are unitary `translation' operators that connect the individual code subspaces, $C^{(\nu)} = T^{(\nu)} C^{(1)}$.  
Interestingly, the orthogonality conditions of Eq.~(\ref{subspaceveccond}) were independently discovered in \cite{grassl2017codes} for this subclass of hybrid subspace codes (see the recent work \cite{majidy2018unification} for further discussions and results on connections between the two perspectives). 
The codes constructed in \cite{grassl2017codes} are captured by the stabilizer formalism presented above; in particular, they are special cases of hybrid codes with (in our notation above) trivial gauge generator group ($\mathcal G_0 = \emptyset$) and non-trivial coset representatives ($\{ I \} \subsetneq \mathcal T_0$). As an illustration, let us consider one of the codes presented there. 

The following describes a single qubit hybrid code on 7-qubit space presented as one of the examples in \cite{grassl2017codes}. The first six rows are the stabilizer subgroup generators, the next two are logical operators on the base code space $C^{(1)}$, and the final row is a translation operator which we shall denote by $T$. 
\[
\begin{array}{c|ccccccccc}
\hline \hline 
S_1 & X & I & I & Z & Y & Y & Z \\ 
S_2 & Z & I & I & I & I & I & X \\ 
S_3 & I & X & I & X & Z & I & I \\ 
S_4 & I & Z & I & Z & I & X & X \\
S_5 & I & I & X & X & I & Z & I \\ 
S_6 & I & I & Z & Z & X & I & X \\ 
\hline 
\overline{X} & I & I & I & X & Z & Z & X \\  
\overline{Z} & I & I & I & Z & X & X & I \\ 
\hline 
T & I & I & I & I & X & Y & Y \\ 
\hline \hline
\end{array}
\]

In our notation, the parameters for this example are $n=7$, $s=6$, and $k=1$ (with $r= 0$ as there is no subsystem structure here). The choice of coset representatives given by the table is $\mathcal T_0 = \{I, T\}$; and indeed, one can see that $T$ does not commute with $S_k$ for $k=2,3,5,6$, so $T \mathcal N(\mathcal S) \neq \mathcal N(\mathcal S)$ is a different coset than that defined by the identity operator.  Observe that there are $2^s = 64$ cosets for $\mathcal N(\mathcal S)$ in this case, and so there are several other potential coset representative subset choices. As a simple example, we could choose $X_1 = XIIIIII$, which does not commute with $S_2$, and so defines a different coset than the identity operator and  that defined by $T$ (as $T X_1 \notin \mathcal N(\mathcal S)$ since it does not commute with $S_3$); that is, $ \mathcal N(\mathcal S) \neq X_1  \mathcal N(\mathcal S) \neq T  \mathcal N(\mathcal S)$.

Additionally, Theorem~\ref{errorcorrecthm} gives a characterization of sets of possible Pauli errors that can be corrected by such codes. Correctable errors can be viewed through this lens; for instance, while the operator $X_1$ can act as a new coset representative for this hybrid code, it is also a correctable error for the code, as can be seen through an application of Theorem~\ref{errorcorrecthm} and the group theoretic conditions displayed there. The relevant operator sets given by the theorem are (recalling that there are no noncommutative gauge operators here): (i) $\mathcal N(\mathcal S) \setminus \langle \mathcal S, iI \rangle $; and, (ii) $T \mathcal N(\mathcal S) = \mathcal N(\mathcal S) T$. The fact that $X_1$ does not belong to either of these sets shows the set $\{ I, X_1 \}$ is a correctable set of errors for the code.    

Note that Theorem~\ref{errorcorrecthm} also tells us what sets of errors are not correctable for this hybrid code. As another simple example, consider the error set $\{ I, T\}$.  Observe that $TI = T\in T  \mathcal N(\mathcal S)$, and so the error set fails the hybrid correctable condition. (In fact, as noted in the previous example, the same is true for any error set consisting of transversal operators.) 
This error set is interesting in the sense that, while it does not satisfy the hybrid correctable condition, each of the individual quantum subspace codes are correctable for the error set, which follows since $T\notin  \mathcal N(\mathcal S) \supseteq \mathcal N(\mathcal S) \setminus \langle \mathcal S, iI \rangle$.

\subsection{Hybrid Bacon-Shor Code}

The (two-dimensional) Bacon-Shor code~\cite{Bacon2006Operator,Li2019Compass} is a subsystem code defined on an $\ell \times \ell$ grid of qubits, with gauge group $\mathcal G$ generated by $\mathcal G_0$ and $iI$, where
\begin{equation}
    \mathcal G_0 = 
    \{
    X_{(i,j)} X_{(i,j+1)} :
    1 \leq i \leq \ell, 1 \leq j < \ell
    \}
    \cup
    \{
    Z_{(i,j)} Z_{(i+1,j)} :
    1 \leq i < \ell, 1 \leq j \leq \ell
    \}.
\end{equation}
We use the notation $X_{(i,j)}$ to denote a Pauli $X$ operator acting on the qubit at coordinate $(i,j)$ (and analogously for Pauli $Z$ operators).
See Fig.~\ref{fig:bacon-shor-lattice} for a visual depiction of the operators in $\mathcal G_0$.
The stabilizer group $\mathcal S$ is generated by the set  
\begin{equation}
    \big\{
    X_{(\ast,j)} X_{(\ast,j+1)}, Z_{(i,\ast)} Z_{(i+1,\ast)} :
    1 \leq i,j < \ell - 1
    \big\} ,
\label{eq:bacon-shor-stab}
\end{equation}
where $X_{(\ast,j)} = X_{(1,j)} X_{(2,j)} \ldots X_{(\ell, j)}$.
The logical group $\mathcal L$ is generated by $\mathcal L_0$ and $iI$, where
\begin{equation}
    \mathcal L_0 = 
    \{
    X_{(\ast,1)}, Z_{(1, \ast)}
    \},
\end{equation}
and so the code distance is $d = \ell$.

\begin{figure}
    \centering
    (a)
    \includegraphics[width=.3\linewidth]{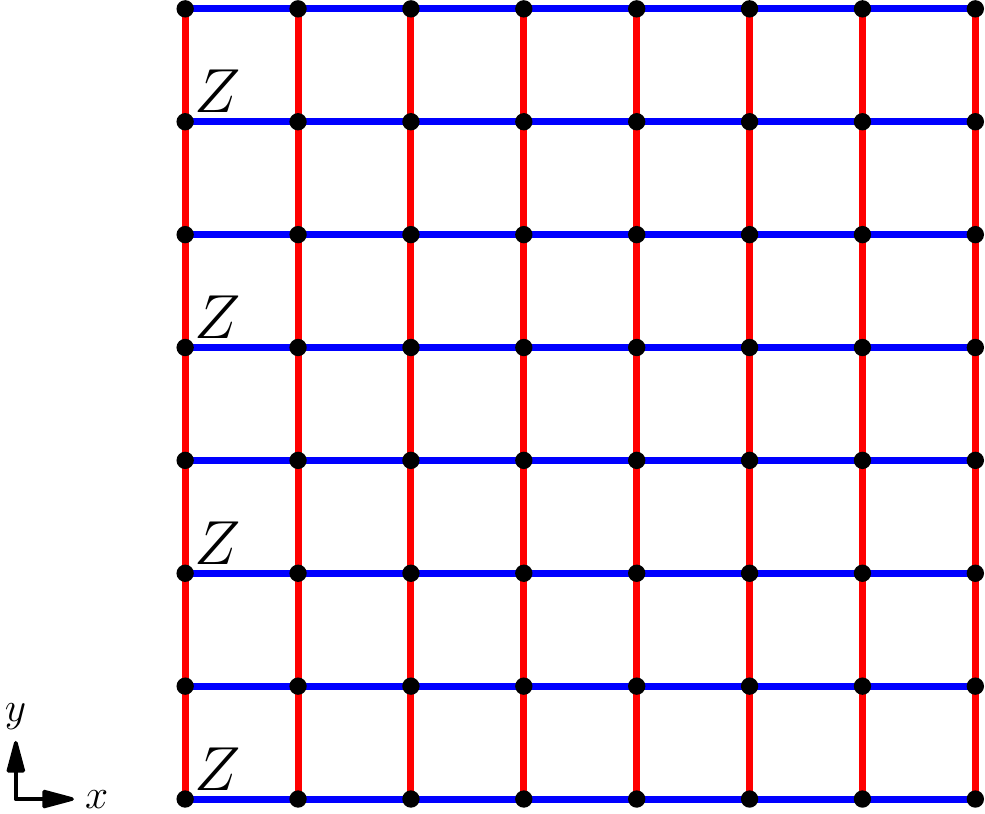}
    \quad\quad
    (b)
    \includegraphics[width=.3\linewidth]{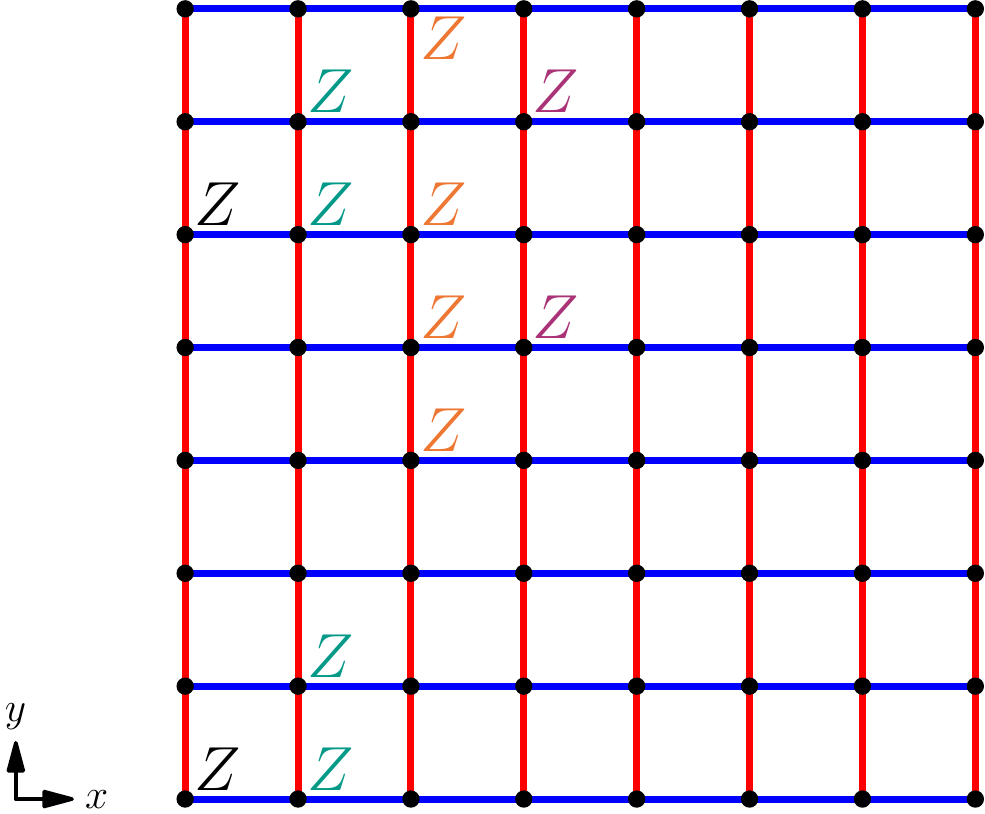}
    \caption{Hybrid Bacon-Shor code.
    Qubits are indicated by black circles, $XX$ gauge generators by red lines and $ZZ$ gauge generators by blue lines.
    (a) Example coset representative whose error syndrome is $11\dots1$.
    (b) Example coset representatives whose error syndromes generate the Hamming code (each colour denotes a different representative).
    }
    \label{fig:bacon-shor-lattice}
\end{figure}

Let us consider some example choices of subsets of coset representatives. 
First, let $\mathcal T_0$ be generated by $\prod_{i=1}^{\lfloor\ell / 2\rfloor} X_{(2i,1)}$ and $\prod_{j=1}^{\lfloor\ell / 2\rfloor} Z_{(1,2j-1)}$ (see Fig.~\ref{fig:bacon-shor-lattice}a for the $\ell=8$ case).
With this choice of $\mathcal T_0$ we get a 2-bit hybrid Bacon-Shor code.
We can equivalently index our subset of coset representatives by their error syndromes. 
Consider the $X$-type stabilizer generators given in Eq.~\eqref{eq:bacon-shor-stab}. 
We write the error syndrome as a binary string, where the $j$'th entry is 0 if the stabilizer $X_{(\ast, j)} X_{(\ast, j+1)}$ is satisfied, and 1 if it is unsatisfied.
Then the error syndromes corresponding to the coset representatives $I$ and $\prod_{j=1}^{\lfloor\ell / 2\rfloor} Z_{(1,2j)}$ are respectively $00\ldots0$ and $11\ldots1$; i.e.\ the codewords of the $(\ell$-1)-bit repetition code.
In general, let $C_c$ be an $(\ell$-1)-bit linear code with basis $\{ v_i \}$. 
We are free to choose $Z$-type generators $g_i \in \mathcal T_0$ such that $\sigma (g_i) = v_i$, where $\sigma (g_i)$ denotes the (binary) error syndrome.
Fig.~\ref{fig:bacon-shor-lattice}b illustrates the case when $C_c$ is the $[7,4,3]$ Hamming code. 
If we make the same choice for the $X$-type generators of $\mathcal T_0$ we get an 8-bit hybrid Bacon-Shor code.
The following theorems characterize the distance of our hybrid Bacon-Shor codes. 

\begin{theorem}
Let $C = C(\mathcal S, \mathcal G_0, \mathcal L_0, \mathcal T_0 = \{I\})$ be an $[\![n,k,d]\!]$ stabilizer subsystem code.
Fix a generating set $\{ S_j : 1 \leq j \leq s \}$ for $\mathcal S$.
Suppose that none of the single-qubit error operators anti-commute with more than $m$ of the $S_j$.
Then for any $[s,k_c,d_c]$ linear code $C_c$, there exists a hybrid subsystem code $C' = C(\mathcal S, \mathcal G_0, \mathcal L_0, \mathcal T_0')$ encoding $k$ logical qubits into $n$ qubits with $|\mathcal T_0'| = 2^{k_c}$ and distance $d' \geq \min(d, \lceil d_c/m \rceil)$.
\label{distancethm}
\end{theorem}

\begin{proof}
Let $\{ v_i \}$ be a basis for the codewords of $C_c$. 
For each $v_i$, we construct a corresponding coset representative $h_i$ such that $\sigma(h_i)=v_i$, giving $\mathcal T_0' = \langle h_i : 1 \leq i \leq k_c \rangle$.
This can be done for example using Gaussian elimination.
Now we apply Theorem~\ref{errorcorrecthm} to $C' = C(\mathcal S, \mathcal G_0, \mathcal L_0, \mathcal T_0')$.
First note that $\mathcal N(\mathcal S)\setminus \mathcal G$ contains operators of weight at least $d$.
To bound the weight of operators in $\bigcup_{i \neq j} g_i \mathcal N(\mathcal S) g_j^{-1}$, we observe that there is a bijection between the $g_i \in \mathcal T'_0$ and the codewords of $C_c$.
As operators in $\mathcal N(\mathcal S)$ have trivial syndrome, every operator of the form $g_i \mathcal N(\mathcal S) g_j^{-1}$ has syndrome equal to a codeword $u = u_i + u_j$ of $C_c$, where $u_i$ and $u_j$ are the codewords of $C_c$ corresponding to $g_i$ and $g_j$, respectively.
Furthermore, $u \neq 0$ as $i \neq j$.
Because any single-qubit error anti-commutes with at most $m$ stabilizer generators, any operator with syndrome equal to $u$ must have weight at least $\lceil d_c/m \rceil$.
\end{proof}

A CSS code~\cite{PhysRevA.54.1098,steane1996multiple} is a stabilizer code whose stabilizer group can be generated by two sets $S_X$ and $S_Z$ consisting of exclusively of $X$-type and $Z$-type operators, respectively.
In this case, we have another option for constructing hybrid codes.

\begin{corollary}
\label{distancecorr}
Let $C = C(\mathcal S, \mathcal G_0, \mathcal L_0, \mathcal T_0 = \{I\})$ be an $[\![n,k,d]\!]$ CSS subsystem code.
Fix $X$-type and $Z$-type generating sets, $S_X$ and $S_Z$, for $\mathcal S$, where $|S_X| = s_X$ and $|S_Z| = s_Z$.
Suppose that none of the single-qubit error operators anti-commute with more than $m_X$ ($m_Z$) operators in $S_X$ ($S_Z$).
Then for any pair of linear codes $C_X$ and $C_Z$ with parameters $[s_X, k_X, d_X]$ and $[s_Z, k_Z, d_Z]$, there exists a hybrid subsystem code $C' = C(\mathcal S, \mathcal G_0, \mathcal L_0, \mathcal T_0')$ encoding $k$ logical qubits into $n$ qubits with $|\mathcal T'_0| = 2^{k_X + k_Z}$ and distance $d' \geq \min(d, d_{XZ})$, where $d_{XZ} = \min(\lceil d_X/m_X \rceil, \lceil d_Z/m_Z \rceil)$.
\end{corollary}

\begin{proof}
Let $\{ v_i \}$ and $\{ w_j \}$ be bases for the codewords of $C_X$ and $C_Z$, respectively.
We construct coset representatives $h_i$ and $f_j$, such that $\sigma(h_i) = (v_i, 0)$ and $\sigma(f_j) = (0, w_j)$, giving $\mathcal T'_0 = \langle h_i, f_j : 1 \leq i \leq k_X, 1 \leq j \leq k_Z \rangle$.
Now we apply Theorem~\ref{errorcorrecthm} to $C' = C(\mathcal S, \mathcal G_0, \mathcal L_0, \mathcal T_0')$.
The only difference with the proof of Theorem~\ref{distancethm} is that every operator in $g_i \mathcal N(\mathcal S) g_j^{-1}$ has syndrome equal to $(u,t)$, where $u$ and $t$ are codewords of $C_X$ and $C_Z$, respectively.
In this case we can have either $u=0$ or $t=0$, but not both.
Because any single-qubit error anti-commutes with at most $m_X$ $X$-type stabilizer generators and $m_Z$ $Z$-type stabilizer generators, any operator with syndrome equal to $(u,t)$ must have weight at least $\min(\lceil d_X/m_X \rceil, \lceil d_Z/m_Z \rceil)$.
\end{proof}

We are now equipped to discuss the distance of our hybrid Bacon-Shor code examples.
We denote the parameters of a hybrid code by $[\![n,k:m,d]\!]$, where $n$ is the number of physical qubits, $k$ is the number of encoded qubits, $m$ is the number of encoded bits, and $d$ is the code distance\footnote{We note that hybrid codes of this form are a subset of our general construction, corresponding to the case where $|\mathcal T_0| = 2^m$.}.
For Bacon-Shor codes with the generating set of the stabilizer group given in \eqref{eq:bacon-shor-stab}, we have $s_X = s_Z = \ell - 1$ and $m_X = m_Z = 2$.
Applying Corollary~\ref{distancecorr}, we find that any hybrid code built from an $[\![\ell^2, 1, \ell]\!]$ Bacon-Shor code using our construction has distance at most $\lceil (\ell - 1)/2 \rceil$.
We can saturate this bound by choosing $C_X$ and $C_Z$ to both be the $[\ell - 1, 1, \ell - 1]$ repetition code, obtaining a hybrid Bacon-Shor code with parameters $[\![ \ell^2, 1:2, \lceil (\ell - 1)/2 \rceil ]\!]$.
This is exactly the case where $T'_0 = \langle \prod_{i=1}^{\lfloor\ell / 2\rfloor} X_{(2i,1)}, \prod_{j=1}^{\lfloor\ell / 2\rfloor} Z_{(1,2j-1)} \rangle$.
For the $\ell = 8$ case with both $C_X$ and $C_Z$ equal to the $[7,4,3]$ Hamming code, we obtain a hybrid Bacon-Shor code with parameters $[\![ 64, 1:8, 2 ]\!]$.
In~\cite{nemec2022encoding}, the authors provide an alternative construction of hybrid codes starting from Bacon-Shor codes.
They obtain hybrid codes with parameters $[\![ \ell^2, 1:(\ell-1)^2, 2 ]\!]$.
In contrast, our hybrid Bacon-Shor codes can have at most $2(\ell - 1)$ encoded bits, but they can have distance up to $\lceil (\ell - 1)/2 \rceil$.

In each of the previous examples, the distance of the hybrid Bacon-Shor code is lower than the distance of the initial Bacon-Shor code.
This need not be the case for all initial codes, e.g., consider the $[\![\ell^2,2,\ell]\!]$ toric code with the canonical stabilizer generators~\cite{Kitaev2003Anyons}.
Here we have $s = \ell^2 - 2$ and $m=4$. 
If we choose a linear code with parameters $[s, \alpha s, \beta s]$ where $\beta \geq 4 \ell / (\ell^2 - 2)$ then, by Theorem~\ref{distancethm}, we can construct a hybrid toric code with parameters $[\![ \ell^2, 2:\alpha(\ell^2-2), \ell]\!]$.

\section{Extension to Qudits}

In this section we discuss the extension of the stabilizer formalism presented above to the case of qudits; that is, what happens when one replaces the base qubit space $\mathbb{C}^2$ with $\mathbb{C}^d$ for fixed positive integer $d > 2$. We begin by recalling the basic set up for the standard qudit stabilizer formalism, as described in several other places (see for instance \cite{gottesman1999fault,gheorghiu2014standard,miller2019small}).

\subsection{The \texorpdfstring{$n$}{n}-Qudit Pauli Group}

Let $\{ \ket{0}, \ldots , \ket{d-1} \}$ be a fixed basis for $\mathbb{C}^d$, and given a fixed positive integer $n\geq 1$ consider the corresponding basis for   
$(\mathbb C^d)^{\otimes n}$ written as $\{ \ket{i_1\cdots i_n} = \ket{i_1}\otimes \ldots \otimes \ket{i_n} \, : \, 0 \leq i_j \leq d-1 , \, 1 \leq j \leq n \}$.
Further let $\omega = e^{2 \pi i / d}$ be a primitive $d$th root of unity, and define the following generalized Pauli operators: 
\[
X = \sum_{k=0}^{d-1} \kb{k+1}{k}  \quad \mathrm{and} \quad Z = \sum_{k=0}^{d-1} \omega^k \kb{k}{k}, 
\]
where in the definition of $X$ we use modulo $d$ arithmetic with $\ket{d}\equiv \ket{0}$. Our choice of generalized Pauli errors is a common one, but there are other choices (for instance error operators related to the finite field with $d$ elements are often used when $d$ is a prime power).  

Some of the relevant properties of the so-called `shift' ($X$) and `clock' ($Z$)  operators include: $X^d = I = Z^d$ and the anti-commutation relation 
\[
\quad Z X = \omega  X Z . 
\]
Note that $X$ and $Z$ are no longer self-adjoint for $d> 2$, but they are unitary with $X^{-1} = X^{d-1} = X^\dagger$ (and the same for $Z$). 
The single qudit Pauli group is the unitary subgroup of $\mathcal U(\mathbb{C}^d)$ given by
\[
\mathcal P_{d,1} = \langle \sqrt{w} I , X , Z \rangle, 
\]
where we use the choice of complex square root $\sqrt{\omega} = e^{\pi i / d}$. So the generic element of $\mathcal P_{d,1}$ can be written in the form $\omega^{a/2} X^b Z^c$ for some $a,b,c\in \mathbb{N}$. 

Observe that for $d=2$ we have $\sqrt{\omega} = i$, and so this definition agrees with the qubit case. But one may ask, why include the phase factor $\sqrt{\omega}$ as a generator, instead of $\omega$ for instance? The reason is that including it allows for many more eigenvalue-1 operators, which is crucial in the context of the stabilizer formalism. Indeed, one can show using standard linear algebra tools that for any operator $X^a Z^b$, with $a,b\in\mathbb{N}$, there is a $U\in\mathcal P_{d,1}$ that is proportional to the operator such that $U$ has 1 as an eigenvalue. 

As in the single qubit case, for arbitrary $n \geq 1$, we define the {\it $n$-qudit Pauli group} $\mathcal P_{d,n}$ to be the subgroup of $\mathcal U(N)$, with $N=d^n$, generated by $n$-tensors of the single qudit Pauli operators $X$, $Z$, and $\sqrt{\omega} I$; that is, the unitary group generated by $Z_1 = Z \otimes (I^{\otimes (n-1)})$, $Z_2 = I\otimes  Z \otimes (I^{\otimes (n-2)})$, etc. Hence it follows, again applying the anti-commutation relations, that a generic element of $\mathcal P_{d,n}$ belongs to the set: 
\[
\Big\{   (\sqrt{\omega})^c  X_{1}^{a_{1}} Z_{1}^{b_{1}} \cdots    X_{n}^{a_{n}} Z_{n}^{b_{n}} \, : \, 0\leq c \leq 2d-1, \,\, 0 \leq a_j, b_j \leq d-1   \Big\}. 
\] 
Observe that the cardinality of $\mathcal P_{d,n}$ is: $2d\times d^n \times d^n = 2 d^{2n+1}$. 

For stabilizer formalism related calculations, it is useful to know that every element of $\mathcal P_{d,n}$ that is not a multiple of the identity operator has trace equal to 0. Indeed, one can use the anti-commutation relation and cyclic property of the trace to show that $\mathrm{Tr} (X^aZ^b) \neq 0$ if and only if $X^a Z^b$ is a multiple of the identity.  
Moreover, given an abelian subgroup $\mathcal S$ of $\mathcal P_{d,n}$, there is a well-known and useful formula for the orthogonal projection $P_{\mathcal S}$ onto the stabilizer subspace defined by $\mathcal S$ given by $P_{\mathcal S} = \frac{1}{|\mathcal S|} \sum_{S\in \mathcal S} S.$

\subsection{Hybrid Qudit Stabilizer Formalism} 

The OAQEC stabilizer formalism presented above for the qubit base space, extends fully to the case of qudits, including the main error correction theorem. Here we briefly point out the main pieces, following along the presentation above. 

\begin{itemize}
\item[$\bullet$] The starting point is again an abelian subgroup $\mathcal S$ of $\mathcal P_{d,n}$, and now we require that $\mathcal S$ contains no scalar operators other than the identity operator $I$ (which is equivalent to $-I\notin \mathcal S$ in the $d=2$ case). Even though the generating operators are no longer self-adjoint, it is still the case that the normalizer and centralizer coincide; that is, $\mathcal N(\mathcal S) = \mathcal Z(\mathcal S)$. (This follows because elements of $\mathcal P_{d,n}$ either commute or commute up to a power of $\omega$, and $I$ is the only scalar operator in $\mathcal S$.) The stabilizer subspace $C = C(\mathcal S) = \mathrm{span} \{ \ket{\psi} \, : \, g \ket{\psi} = \ket{\psi} \,\, \forall g\in \mathcal S \}$ is defined in the same way. 

\item[$\bullet$] The $r$-qudit gauge group and $k$-qudit logical group are analogously defined, with $\sqrt{\omega}I$ replacing $iI$. Lemma~\ref{reducinglemma} holds, with $\mathbb{C}^d$ replacing $\mathbb{C}^2$ in the Hilbert space and subsystem decompositions. 
  
\item[$\bullet$] Further, regarding the normalizer cosets and hybrid code sectors, Lemma~\ref{cosetlemma} still holds, with the replacement of $\mathcal P_{d,n}$ and $\omega I$ in the statement (with the same basic ingredients in the proof, as noted just above). Thus, given a subset of a coset transversal for $\mathcal N(\mathcal S)$ inside $\mathcal P_{d,n}$, we will have an associated hybrid code $C = C(\mathcal S, \mathcal G_0, \mathcal L_0, \mathcal T_0)$ with code sectors as in the qubit case, and subsystem structure defined by the gauge group (when it is non-trivial), which is carried to the different code spaces by the transversal operators. 
\end{itemize}

In terms of the error correction conditions, first note that none of the OAQEC framework or results are qubit dependent, they are based on the general theory of operator algebras on Hilbert space. This, together with using analogous properties of the generalized Pauli group, allows us to generalize Theorem~\ref{errorcorrecthm}, essentially with the same proof. We state the result here for completeness.  

\begin{theorem} \label{errorcorrecthmd}
A code $C = C(\mathcal S, \mathcal G_0, \mathcal L_0, \mathcal T_0)$, with $\mathcal T_0 = \{ g_i \}_i$,  is correctable for a set of error operators $\{E_k\} \subseteq \mathcal P_{d,n}$ if and only if for all $k,l$,  
\begin{equation}\label{stabformoaqeccondd}
E_k^\dagger E_l \notin \Big( \mathcal N(\mathcal S)\setminus \mathcal G  \Big) \bigcup \Big(  \bigcup_{i\neq j} g_i \mathcal N(\mathcal S) g_j^{-1}   \Big) .
\end{equation}
\end{theorem}

A pair of examples are discussed below. 

\begin{example} 
The motivating example presented above generalizes straightforwardly as follows. 
\begin{itemize}
\item[$\bullet$] Let $s \leq n$ be a fixed positive integer and let $\mathcal S = \langle Z_1, \ldots , Z_s \rangle \subseteq \mathcal P_{d,n}$.
Then 
\[
C = C(\mathcal S) = \mathrm{span} \big\{ \ket{ \underbrace{0 \cdots 0}_{s} \, i_1 \cdots i_{n-s}  } \, : \, 0 \leq i_j \leq d-1 \big\}, 
\]
and $\dim C = d^{n-s}$ so that $C$ can encode $n-s$ qudits. 

\item[$\bullet$] Let $r$ be a fixed integer with $0 \leq r \leq n-s$, and let $\mathcal G_0$ be the set of $r$ pairs of generating Pauli operators acting on qudits $s+1$ to $r+s$: 
\[
\mathcal G_0 = \big\{  X_i, Z_i \, : \, s+1 \leq i \leq r+s \big\}. 
\]
Then the gauge group $\mathcal G$ is generated by $\mathcal S$, $\sqrt{\omega} I$, and $\mathcal G_0$, and includes the full subgroup of operators in $\mathcal P_{d,n}$ acting non-trivially on the $r$ gauge qudits.  

\item[$\bullet$] Let $k = n - s - r$,  and let $\mathcal L_0$ be the set of $k$ pairs of generating Pauli operators acting on qudits $r+s+1$ to $n$:
 \[
\mathcal L_0 = \big\{  X_i, Z_i \, : \, r+s+1 \leq i \leq n \big\}. 
\]
Then the logical group $\mathcal L$ is the group generated by $\mathcal L_0$ and $\sqrt{\omega} I$, and includes the full subgroup of operators in $\mathcal P_{d,n}$ acting non-trivially on the $k$ logical qudits. 

\item[$\bullet$] The normalizer $\mathcal N(\mathcal S) = \mathcal Z(\mathcal S)$ for $\mathcal S$ inside $\mathcal P_{n,d}$ is given by: 
\[
\mathcal N(\mathcal S) = \Big\{   \omega^{c/2} \, \cdot \, Z_1^{b_1} \cdots Z_s^{b_s} \, \cdot \,  X_{s+1}^{a_{s+1}} Z_{s+1}^{b_{s+1}} \cdots    X_{n}^{a_{n}} Z_{n}^{b_{n}} \, : \, 0\leq c \leq 2d-1, \,\, 0 \leq a_j, b_j \leq d-1   \Big\}. 
\]
\end{itemize}

The size of the normalizer here is: $|\mathcal N(\mathcal S)| = 2d \times d^s \times (d^2)^{n-s} = 2 d^{2n -s + 1}$. Hence the number of normalizer cosets is given by, 
\[
|\mathcal P_{n,d}| / |\mathcal N(\mathcal S)| = d^s. 
\]
As in the qubit case, each operator $X_j$, $1\leq j \leq s$, does not belong to $\mathcal N(\mathcal S)$, and nor does any product $X_j X_{j'}^{-1}$  of operators from this set. So a coset transversal of maximal size is given by the $d^s$-element set: 
\[
\mathcal T = \big\{ X_1^{a_1} \cdots X_s^{a_s} \, : \, 0 \leq a_j \leq d-1 \big\}.  
\]

Thus, given a (non-trivial) subset of coset representatives $\mathcal T_0 \subseteq \mathcal T$, the subspaces $T C$, $T\in\mathcal T_0$, are mutually orthogonal and the corresponding subspace for the hybrid code is $C_{\mathcal T_0} = \oplus_{T\in\mathcal T_0} T C$, with $C$ (in the case of a non-trivial gauge group) having subsystem structure that is carried to the subspaces $TC$ by the transversal operators. 
A full characterization of the possible correctable errors for any given coset subset $\mathcal T_0$ is given by Theorem~\ref{errorcorrecthmd}. The simple example of a two element transversal set and set of correctable errors discussed above in the qubit case, carries through analogously.  
\end{example} 

As another example, we give a hybrid version of a (subspace) code presented in the seminal work \cite{gottesman2001encoding}. In contrast to the qubit ($d=2$) case, this example shows how for larger $d$, even a single mode ($n=1$) can generate interesting hybrid code structures.  

\begin{example}
Let $d=18$ and $n=1$. A single qubit subspace code, which can be viewed as a `pre-GKP code', is given in \cite{gottesman2001encoding} as the stabilizer subspace $C$ generated by the (abelian) group $\mathcal S = \langle X^6, Z^6 \rangle$, which one can readily calculate is spanned by the two states: 
\[
\ket{\overline{0}} = \frac{1}{\sqrt{3}} (\ket{0} + \ket{6} + \ket{12} ) \quad\quad \mathrm{and} \quad\quad  \ket{\overline{1}} = \frac{1}{\sqrt{3}} (\ket{3} + \ket{9} + \ket{15} ) . 
\]
(Note that for general $d$, the dimension of $\mathcal C$ need not be a power of $d$, as this example shows.) The commutation relations as it relates to these two operators are given, for any positive integers $a$, $b$, as follows: 
\begin{eqnarray*}
(X^a Z^b) X^6 &=& \omega^{6b} X^6 (X^a Z^b) \\ 
(X^a Z^b) Z^6 &=& \overline{\omega}^{6a} Z^6 (X^a Z^b) . 
\end{eqnarray*}
In particular, $X^a Z^b$ commutes with $\mathcal S$ if and only if $a$ and $b$ are both divisible by 3. Thus, in this case we have 
\[
\mathcal N(\mathcal S) = \mathcal Z(\mathcal S) = \big\{ \omega^{c/2} X^a Z^b \,\,\, \big| \,\, \, 0\leq c \leq 35, \, a,b\in \{0,3,6,9,12,15\} \big\}. 
\]
Hence, $|\mathcal N(\mathcal S)| = 36 \times 6 \times 6$, and from $|\mathcal P_{18,1} | = 36 \times 18 \times 18$  it follows that the number of cosets is given by:  $|\mathcal P_{18,1} | / |\mathcal N (\mathcal S) | = 9$. 

Returning to the original code construction, the logical operators were identified as $\overline{X} = X^3$ and $\overline{Z} = Z^3$. Moreover, as noted in \cite{gottesman2001encoding}, the 9 operators belonging to the set $\mathcal T = \{ X^a Z^b \, : \, |a|,|b|\leq 1 \}$ form a correctable error set for the code (in the classic Knill-Laflamme sense, which remember is captured as a special case of OAQEC), as one can show they map $C$ to 9 mutually orthogonal subspaces. This set of operators is also of interest here, as $\mathcal T$ forms a coset transversal for $\mathcal S$. Indeed, one can easily verify using the anti-commutation relations that any two elements from this set define different cosets for $\mathcal S$ inside $\mathcal P_{18,1}$ (as any product $T_1^{-1} T_2 \notin \mathcal N(\mathcal S)$ when $T_1,T_2\in \mathcal T$). 

We can thus consider hybrid versions of this code, by taking a subset $\mathcal T_0$ of elements from $\mathcal T$ and their corresponding code sectors, and then Theorem~\ref{errorcorrecthmd} can tell us what are the correctable error sets for the code. Consider for example the set $\mathcal T_0 = \{I, X, X^{-1} \} \subseteq \mathcal T$. Here the gauge group is generated by $\mathcal S$ and $\sqrt{\omega}I$, and for this particular $\mathcal T_0$ the first set in the union of Eq.~(\ref{stabformoaqeccondd}) is equal to: 
\[
\mathcal N(\mathcal S) \setminus \mathcal G =  \big\{ \omega^{c/2} X^a Z^b \,\,\, \big| \,\, \, 0\leq c \leq 35, \, a,b\in \{3,9,15\} \big\}, 
\]
which follows from elements commuting modulo a power of $\omega$. 
The cross-term union of Eq.~(\ref{stabformoaqeccondd}) is defined by 6 sets, but due to repeats it collapses to the union of 4 sets given by $\bigcup_{a\in \{ -2, -1, 1, 2\}} X^a \mathcal N(\mathcal S)$, which one can check is equal to the set:  
\[
\big\{ \omega^{c/2} X^a Z^b \,\,\, \big| \,\, \, 0\leq c \leq 35, \, a\in \{ 1,2,4,5,7,8,10,11,13,14,16,17  \}, \, b\in \{0,3,6,9,12,15\} \big\}. 
\] 
Observe the first set is a subset of the second union, which is a consequence of $\mathcal T_0$ including $I$ and being inverse closed. 
Thus, by Theorem~\ref{errorcorrecthmd}, the possible correctable errors for this code are precisely those operator sets $\mathcal E \subseteq \mathcal P_{18,1}$ such that $g_1^{-1} g_2$ does not belong to this union for any choice of $g_1,g_2\in\mathcal E$. For example, one can easily check that the set $\mathcal E = \{ Z^{2b+1} \, : \, 0\leq b \leq 8 \}$ satisfies this condition and hence forms a correctable set of errors for the code. 

Note that the hybrid code in this particular instance is 6-dimensional, as it is determined by a qubit base code and the 3 code sectors defined by $\mathcal T_0 = \{I, X, X^{-1} \}$; namely, the direct sum  $C_{\mathcal T_0} = C \oplus XC \oplus X^{-1}C$. Thus, as we are in a 18-dimensional space, we can have a maximum of 3 (non-degenerate) errors that can be correctable for this code. One might express concern then, as the error set $\mathcal E$ includes 9 operators; however, there is no contradiction here, as these operators include degeneracy on $C_{\mathcal T_0}$. Indeed, one can check directly that $Z, Z^3, Z^5$ map this 6-dimensional subspace to 3 mutually orthogonal subspaces. Moreover, $Z^6$ acts as the identity on $C$, and from the anti-commutation relations it acts as $\omega^6I$ on $XC$ and $\overline{\omega}^6I$ on $X^{-1}C$. It follows that $Z C_{\mathcal T_0}  = Z^7  C_{\mathcal T_0} = Z^{13}  C_{\mathcal T_0}$, where this is equality of subspaces, and analogous statements are true for the operator triples $\{Z^3, Z^9, Z^{15}\}$ and  $\{Z^5, Z^{11}, Z^{17}\}$  on the other two (orthogonal) 6-dimensional subspaces defined by $Z^3$ and $Z^5$. So the 9 operator error set actually degenerates in this case to 3 different errors when one restricts to the hybrid code space. 
\end{example}


\section{Concluding Remarks}

This work opens up a number of potential new lines of investigation and the possible extension of some others. 
Further consideration of the hybrid Bacon-Shor subsystem codes introduced here is warranted, given the wide applicability of the subsystem versions \cite{Poulin2005Stabilizer,Bacon2006Operator} in fault-tolerant quantum computing and beyond, and in particular with NISQ era quantum computers likely to involve hybrid forms of classical and quantum information processing \cite{preskill2018quantum}. 
It would be interesting to explore possible implications of our stablizer formalism on other classes of recently constructed hybrid codes; for instance, we expect new light can be shed on codes constructed in works such as \cite{grassl2017codes,li2020error,nemec2021infinite,nemec2022encoding}, and one can ask if the formalism allows for construction of more codes with useful properties following the approaches introduced there.  
One could also consider generalizations of this formalism to a variety of other settings of relevance in quantum information, such as other generalized Pauli error models, continuous QEC and infinite-dimensional settings such as in \cite{gottesman2001encoding}, or entanglement-assisted error correction \cite{brun2006correcting,hsieh2007general,kremsky2008classical}. 
Regarding the connection with black hole theory, it may be possible to use our OAQEC stabilizer formalism to construct toy models of AdS/CFT capturing the properties missed by the celebrated tensor-network models of~\cite{Pastawski2015,Hayden2016}, which are subsystem codes. We plan to pursue these investigations elsewhere and we invite others to do so as well.

\strut

{\noindent}{\it Acknowledgements.}
We dedicate this paper to the memory of our friend and mentor, David Poulin. We thank Priya Nadkarni and Rafael Alexander for stimulating discussions, and to Tarik El-khateeb for help with the hybrid code figure design.  We are grateful to the referees for several helpful comments and suggestions. D.W.K. was partly supported by NSERC Discovery Grant 400160. Research at Perimeter Institute is supported in part by the Government of Canada through the Department of Innovation, Science and Economic Development Canada and by the Province of Ontario through the Ministry of Colleges and Universities.

\bibliographystyle{plainurl}

\bibliography{refs}

\end{document}